\documentclass{article}

\usepackage{amsfonts}
\usepackage{amsmath}
\usepackage{amssymb}
\usepackage{amsthm}

\usepackage[final,letterpaper]{hyperref}
\usepackage{tikz}

\hypersetup{
pdfauthor = {Jack H. Lutz, Neil Lutz},
pdftitle = {Lines Missing Every Random Point},
pdfsubject = {Lines Missing Every Random Point},
pdfkeywords = {randomness, algorithmic geometric measure theory, computable analysis},
pdfcreator = {LaTeX with hyperref package},
pdfproducer = {dvips + ps2pdf}}

\usepackage{amsfonts}
\usepackage{amsmath}
\usepackage{amssymb}
\usepackage{enumerate}

\newcommand{\mb}{\mathbf}
\newcommand{\wh}{\widehat}

\newcommand{\N}{\mathbb{N}}
\newcommand{\Z}{\mathbb{Z}}
\newcommand{\Q}{\mathbb{Q}}
\newcommand{\R}{\mathbb{R}}
\newcommand{\Ref}{\mathrm{Ref}}
\newcommand{\Rot}{\mbox{Rot}}
\newcommand{\MartinLof}{Martin-L\"{o}f}
\newtheorem{thm}{Theorem}
\newtheorem{coro}[thm]{Corollary}
\newtheorem{obser}[thm]{Observation}
\newtheorem{lemm}[thm]{Lemma}
\newtheorem{prop}[thm]{Proposition}

\begin{document}
\title{\bf{Lines Missing Every Random Point}\footnote{A preliminary version of a portion of this work was presented at the Conference on Computability in Europe (Budapest, Hungary, June 23-27, 2014).}}
\author{
Jack H. Lutz\footnote{Department of Computer Science, Iowa State University, Ames, IA 50011, USA. {\tt lutz@cs.iastate.edu}. This author's research was supported in part by National Science Foundation Grant 1247051.}
\and
Neil Lutz\footnote{Department of Computer Science, Rutgers University, Piscataway, NJ 08854, USA. {\tt njlutz@cs.rutgers.edu}. This author's research was supported in part by National Science Foundation Grant 1101690.}
}

\date{}

\maketitle

\begin{abstract}
We prove that there is, in every direction in Euclidean space, a line that misses every computably random point. We also prove that there exist, in every direction in Euclidean space, arbitrarily long line segments missing every double exponential time random point.
\end{abstract}

\begin{section}{Introduction}\label{se:1}
One objective of the theory of computing is to investigate the fine-scale geometry of algorithmic information in Euclidean space. Recent work along these lines has included algorithmic classifications of points lying on computable curves and arcs~\cite{GuLuMa06, RetZhe09, CoDaMc12, ZheRet12, McNi13} and in more exotic sets~\cite{LutMay08, KjoNer09, DLMT14, GLMM15}.

This paper concerns a simple, fundamental question: Can the direction of a line in Euclidean space force the line to meet at least one random point? That is, can the set of \MartinLof{} random points, which is everywhere dense and contains almost every point in Euclidean space, be avoided by lines in every direction? For example, it is reasonable to conjecture that every line of random slope in $\R^2$ contains a random point. We show here that this conjecture is false, and in fact that---regardless of slope---every line can be translated so that it contains no \MartinLof{} random point. Moreover, the line can miss the larger class of all computably random points.

Our solution of this problem builds on a very old---and ongoing---line of research in geometric measure theory. In 1917 Fujiwara and Kakeya~\cite{Kake17,FujKak17} posed the question of the minimum area of a plane set in which a unit segment can be continuously reversed without leaving the set, a \emph{Kakeya needle set}. This question was resolved in 1928 by Besicovitch~\cite{Besi28b}: such a set can have arbitrarily small measure. The work made use of a construction by Besicovitch from 1919~\cite{Besi19} (but not widely circulated until its republication in 1928~\cite{Besi28a}) of a plane set of area $0$ containing a unit line segment in every direction, a \emph{Kakeya set}. This set was constructed using a clever iterated process of partitioning and translating the pieces of an equilateral triangle.

In 1964 Besicovitch used a duality principle to construct a plane set with area 0 that contains a line in every direction, a \emph{Besicovitch set}~\cite{Besi64}. Falconer~\cite{Falc80,Falc85} used an alternative duality principle to give a somewhat simpler construction of a Besicovitch set. This latter set $B$, which is the point-line dual of a simply defined ``fractal dust,'' is described in detail in Section \ref{se:4}. Our main result is achieved by showing that $B$ has computable measure $0$, as does its Cartesian product with $\R^n$, for every $n\in\N$. We also sketch an alternative proof suggested to us by Turetsky (personal communication) and an anonymous reviewer.

Our main result leads us to conjecture that there is, in every direction in Euclidean space, a line that misses not only every computably random point, but every point that is feasibly random (i.e., polynomial time random, as defined in Section \ref{se:2}). We are unable to prove this conjecture at this time, but in Section \ref{se:5} we prove a weaker result along these lines. Specifically, we show that there exist, in every direction in the Euclidean plane, arbitrarily long line segments missing every point that is double exponential time random (a randomness condition defined in Section \ref{se:2}). Our proof of this fact uses Besicovitch's above-mentioned 1919 construction of a Kakeya set, together with later refinements of this proof by Perron~\cite{Perr28}, Schoenberg~\cite{Scho62}, and Falconer~\cite{Falc85}.

More recent work on the ``sizes'' of Besicovitch sets and Kakeya sets has focused on their dimensions. Davies showed that every Kakeya set in $\R^2$ has Hausdorff dimension $2$~\cite{Davi71}, and the famous Kakeya conjecture states that Kakeya sets in $\R^n$ have Hausdorff dimension $n$ for all $n\geq2$. For more on this history, consult~\cite{Falc85, KatTao02}.

The remainder of the paper is organized as follows. Section \ref{se:2} contains preliminary information regarding computable and time-bounded measure and randomness in $\R^n$. In Section \ref{se:3}, we present a class of martingales for betting on open sets. In Section \ref{se:4}, we describe Falconer's Besicovitch set $B$ and prove the main theorem in $\R^2$. In Section \ref{se:5} we describe a Kakeya set $K$ and use it to prove our result on segments missing every double exponential time random point in $\R^2$. Section \ref{se:6} extends our two theorems to $\R^n$ ($n\geq2$). Section \ref{se:7} mentions open problems.
\end{section}

\begin{section}{Computable and Time-Bounded Randomness in $\R^n$}\label{se:2}
We now discuss the elements of computable measure and randomness in $\R^n$. For each $r\in\N$ and each $\mb{u}=\left(u_1,...,u_n\right)\in\Z^n$, let
\[Q_r\left(\mb{u}\right)=\left[u_1\cdot2^{-r},\left(u_1+1\right)\cdot2^{-r}\right)\times...\times\left[u_n\cdot2^{-r},\left(u_n+1\right)\cdot2^{-r}\right)\]
be the $r$-\emph{dyadic cube} at $\mb{u}$. Note that each $Q_r\left(\mb{u}\right)$ is ``half-open, half-closed'' in such a way that, for each $r\in\N$, the family
\[\mathcal{Q}_r=\bigl\{Q_r\left(\mb{u}\right)\;\big|\;\mb{u}\in\{0,...,2^r-1\}^n\bigr\}\]
is a partition of the unit cube $Q_0\left(\mb{0}\right)=\left[0,1\right)^n$. The family
\[\mathcal{Q}=\bigcup_{r=0}^\infty\mathcal{Q}_r\]
is the set of all \emph{dyadic cubes} in $\left[0,1\right)^n$.

A \emph{martingale} on $\left[0,1\right)^n$ is a function $d:\mathcal{Q}\to\left[0,\infty\right)$ satisfying
\begin{equation}\label{eq:1}
d\left(Q_r\left(\mb{u}\right)\right)=2^{-n}\sum_{\mb{a}\in\{0,1\}^n}d\left(Q_{r+1}\left(2\mb{u}+\mb{a}\right)\right)
\end{equation}
for all $Q_r\left(\mb{u}\right)\in\mathcal{Q}$. Intuitively, a martingale $d$ is a strategy for placing successive bets on the location of a point $\mb{x}\in\left[0,1\right)^n$. After $r$ bets have been placed, the bettor's capital is
\[d^{\left(r\right)}\left(\mb{x}\right)=d\left(Q_r\left(\mb{u}\right)\right)\;,\]
where $\mb{u}$ us the unique element of $\{0,...,2^r-1\}^n$ such that $\mb{x}\in Q_r\left(\mb{u}\right)$. The bettor's next bet is on which of the $2^n$ immediate subcubes $Q_{r+1}\left(2\mb{u}+\mb{a}\right)$ of $Q_r\left(\mb{u}\right)$ has $\mb{x}$ as an element. The condition \eqref{eq:1} says that the bettor's expected capital after this bet is exactly the bettor's capital before the bet, i.e., the payoffs are fair. A martingale $d$ \emph{succeeds} at a point $\mb{x}\in\left[0,1\right)^n$ if
\[\limsup_{r\to\infty}d^{\left(r\right)}\left(\mb{x}\right)=\infty\;.\]
A well known theorem of Ville~\cite{Vill39}, restated in the present setting, says that a set $E\subseteq\left[0,1\right)^n$ has Lebesgue measure $m(E)=0$ if and only if there is a martingale $d$ that succeeds at every point $\mb{x}\in E$. It follows easily by the countable additivity and translation invariance of Lebesgue measure that a set $E\subseteq\R^n$ has Lebesgue measure $0$ if and only if there is a martingale $d$ that succeeds at every point $\mb{x}\in E^\#$, where
\begin{equation}\label{eq:2}
E^\#=\left[0,1\right)^n\cap\bigcup_{\mb{t}\in\Z^n}\left(E+\mb{t}\right)\;.
\end{equation}

Let
\[J=\bigl\{\left(r,\mb{u}\right)\in\N\times\Z^n\;\big|\;\mb{u}\in\{0,...,2^r-1\}^n\bigr\}\;.\]
Then a martingale $d:\mathcal{Q}\to\left[0,\infty\right)$ is \emph{computable} if there is a computable function $\widehat{d}:\N\times J\to\Q\cap\left[0,\infty\right)$ such that, for all $\left(s,r,\mb{u}\right)\in\N\times J$,
\begin{equation}\label{eq:3}
\left|\widehat{d}\left(s,r,\mb{u}\right)-d\left(Q_r\left(\mb{u}\right)\right)\right|\leq2^{-s}\;.
\end{equation}
A set $E\subseteq\R^n$ is defined to have \emph{computable measure} $0$ if there is a computable martingale $d$ that succeeds at every point $\mb{x}\in E^\#$, where $E^\#$ is defined as in \eqref{eq:2}. A point $\mb{x}\in\R^n$ is \emph{computably random} if it is not an element of any set of computable measure $0$, i.e., if there is no computable martingale that succeeds at $\mb{x}$. Computable randomness was introduced by Schnorr~\cite{Schn71a,Schn71b}. It is well known~\cite{Nies09, DowHir10} that every random point in $\R^n$ (i.e., every \MartinLof{} random point in $\R^n$) is computably random and that the converse does not hold. In particular, then, almost every point in $\R^n$ is computably random.

Resource-bounded measure, a complexity-theoretic generalization of Lebesgue measure that induces measure on complexity classes, has been used to define complexity-theoretic notions of randomness~\cite{Lutz92}. Adapting these notions to Euclidean space, a martingale $d:\mathcal{Q}\to[0,\infty)$ is p-\emph{computable} (respectively, ee-\emph{computable}) if there is a function $\widehat{d}:\N\times J\to\Q\cap[0,\infty)$ that satisfies (\ref{eq:3}) and is computable in $(s+r)^{O(1)}$ time (respectively, in $2^{2^{O(s+r)}}$ time). A point $\mb{x}\in\R^n$ is p-\emph{random} (or \emph{polynomial time random}, or \emph{feasibly random}) if no p-computable martingale succeeds at $\mb{x}$~\cite{Lutz92}. A point $\mb{x}\in\R^n$ is ee-\emph{random} (or \emph{double exponential time random}) if no ee-computable martingale succeeds at $\mb{x}$~\cite{HarHit07}. It is routine to show that every computably random point is ee-random, that every ee-random point is p-random, and that the converses of these statements are false.
\end{section} 

\begin{section}{Betting on Open Sets}\label{se:3}
In this section we describe a class of martingales that are used in the proof of the main theorem in Section \ref{se:4}. These martingales are also likely to be useful in future investigations.

For any set $G\subseteq\left[0,1\right)^n$ with $m\left(G\right)>0$, define a martingale $d_G:\mathcal{Q}\to\left[0,\infty\right)$ recursively as follows.
\begin{enumerate}[(i)]
\item $d_G\left(Q_0\left(\mb{0}\right)\right)=1.$
\item For all $r\geq0$, $\mb{u}\in\{0,...,2^r-1\}^n$, and $\mb{a}\in\{0,1\}^n$
\end{enumerate}
\[d_G\left(Q_{r+1}\left(2\mb{u}+\mb{a}\right)\right)=\left\{\begin{array}{ll}
0&\mbox{ if }d_G\left(Q_r\left(\mb{u}\right)\right)=0\\
2^nd_G\left(Q_r\left(\mb{u}\right)\right)\frac{m\left(G\cap Q_{r+1}\left(2\mb{u}+\mb{a}\right)\right)}{m\left(G\cap Q_r\left(\mb{u}\right)\right)}&\mbox{ otherwise}\;.
\end{array}\right.\]
That is, for each cube $Q\in\mathcal{Q}_r$, the values of the martingale on the immediate subcubes of $Q$ are proportional to the measures of the subcubes' intersections with $G$. 
\begin{thm}\label{thm:1}
For every nonempty set $G$ that is open as a subset of the subspace $\left[0,1\right)^n$ of $\R^n$ and every $\mb{x}\in G$, $d^{\left(r\right)}_G\left(\mb{x}\right)=1/m\left(G\right)$ for all sufficiently large $r$.
\end{thm}
\begin{proof}
Let $G$ be a nonempty open set in the subspace $[0,1)^n$ of $\R^n$. Then $m(G)>0$, and by a routine induction argument, for any $r\in\N$ and $\mb{u}\in\{0,...,2^r-1\}^n$,
\begin{equation}\label{eq:closeddg}
d_G(Q_r(\mb{u}))=2^{nr}\frac{m(G\cap Q_r(\mb{u}))}{m(G)}.
\end{equation}
For any $\mb{x}\in G$ there exists $\varepsilon>0$ such that $\mathcal{B}_\varepsilon(\mb{x})\subseteq G\cap[0,1)^n$. Let $r>-\log(\varepsilon)$ and $Q\in\mathcal{Q}_r$ such that $\mb{x}\in Q$. Then $d^{(r)}_G(\mb{x})=d_G(Q)$, and $2^{-r}<\varepsilon$, so $Q\subseteq\mathcal{B}_\varepsilon(\mb{x})\subseteq E$. Applying \eqref{eq:closeddg},
\[d^{(r)}_G=2^{nr}\frac{m(G\cap Q)}{m(G)}=2^{nr}\frac{m(Q)}{m(G)}=\frac{1}{m(G)}.\]
\end{proof}
When $G$ is open, we call $d_G$ the \emph{open set martingale} for $G$.
\end{section} 

\begin{section}{Betting on a Besicovitch Set}\label{se:4}
This section reviews Falconer's construction of the Besicovitch set $B$ mentioned in the introduction and proves that the set $B$ in fact has computable measure $0$. Hence $B$ contains a line in every direction in $\R^2$, and each of these lines misses every computably random point in $\R^2$.

For each $m,b\in\R$, let $\mathcal{L}_{m,b}\subseteq\R^2$ be the line with slope $m$ and $y$-intercept $b$. Falconer defined the \emph{line set operator} $\mathcal{L}:\mathcal{P}\left(\R^2\right)\to\mathcal{P}\left(\R^2\right)$
by
\[\mathcal{L}\left(F\right)=\bigcup\left\{\mathcal{L}_{m,b}\;|\;\left(m,b\right)\in F\right\}\]
for all $F\subseteq\R^2$. We call $\mathcal{L}\left(F\right)$ the \emph{line set} of $F$. It is easy to verify that the operator $\mathcal{L}$ is monotone and maps compact sets to closed sets.

We are interested in the line set of a particular self-similar fractal $F$, which we now define. Consider the alphabet $\Sigma=\{0,1,2,3\}$. For each $i\in\Sigma$ define the contraction $S_i:\R^2\to\R^2$ by
\[S_i\left(x,y\right)=\frac{1}{4}\left(\left(x,y\right)+\left(i,a_i\right)\right)\;,\]
where $a_0=2$, $a_1=0$, $a_2=3$, and $a_3=1$. For each $w\in\Sigma^*$ define the set $F\left(w\right)\subseteq\R^2$ by the recursion
\begin{align*}
F\left(\lambda\right)&=\left[0,1\right]^2,\\
F\left(iw\right)&=S_i\left(F\left(w\right)\right)
\end{align*}
for all $i\in\Sigma$ and $w\in\Sigma^*$. For each $k\in\N$ let
\[F_k=\bigcup\left\{F\left(w\right)\;|\;w\in\Sigma^k\right\}\;.\]

The sets $F_0$ and $F_1$, along with their line sets, are depicted in Figure \ref{fig:1}. We are interested in the set
\[F=\bigcap_{k=0}^\infty F_k\;.\]
This set $F$ is an uncountable, totally disconnected set, informal called a ``fractal dust.'' More formally it is the \emph{attractor} of the \emph{iterated function system} $\left(S_0,S_1,S_2,S_3\right)$, i.e., it is a \emph{self-similar fractal}.
\begin{figure}
\centering
\parbox{1in}
{
\begin{tikzpicture}[scale=0.3]
\path[fill=lightgray] (0,0) -- (0,4) -- (4,4) -- (4,0);
\path[fill=darkgray] (1,0) -- (1,1) -- (2,1) -- (2,0);
\path[fill=darkgray] (3,1) -- (3,2) -- (4,2) -- (4,1);
\path[fill=darkgray] (2,3) -- (2,4) -- (3,4) -- (3,3);
\path[fill=darkgray] (0,2) -- (0,3) -- (1,3) -- (1,2);
\draw [<->] (0,-2) -- (0,6);
\draw [<->] (-2,0) -- (6,0);
\end{tikzpicture}
}\qquad\qquad
\parbox{2.5in}{
\begin{tikzpicture}[scale=1.2]
\path[fill=lightgray] (-3,1) -- (0,1) -- (3,4) -- (3,0) -- (0,0) -- (-3,-3);
\path[fill=darkgray] (0,0) -- (3,0.75) -- (3,1.5) -- (0,0.25);
\path[fill=darkgray] (0,0.25) -- (3,2.5) -- (3,3.25) -- (0,0.5);
\path[fill=darkgray] (0,0.5) -- (3,0.5) -- (3,1.25) -- (0,0.75);
\path[fill=darkgray] (0,0.75) -- (3,2.25) -- (3,3) -- (0,1);
\path[fill=darkgray] (0,1) -- (-3,0.25) -- (-3,-0.5) -- (0,0.75);
\path[fill=darkgray] (0,0.75) -- (-3,-1.5) -- (-3,-2.25) -- (0,0.5);
\path[fill=darkgray] (0,0.5) -- (-3,0.5) -- (-3,-0.25) -- (0,0.25);
\path[fill=darkgray] (0,0.25) -- (-3,-1.25) -- (-3,-2) -- (0,0);
\draw [<->] (0,4) -- (0,-3);
\draw [<->](-3,0) -- (3,0);
\end{tikzpicture}
}
\caption{$F_0$ and $F_1$, along with their line sets. $F_0$ and $\mathcal{L}\left(F_0\right)$ are shaded gray; $F_1$ and $\mathcal{L}\left(F_1\right)$ are black.}
\label{fig:1}
\end{figure}
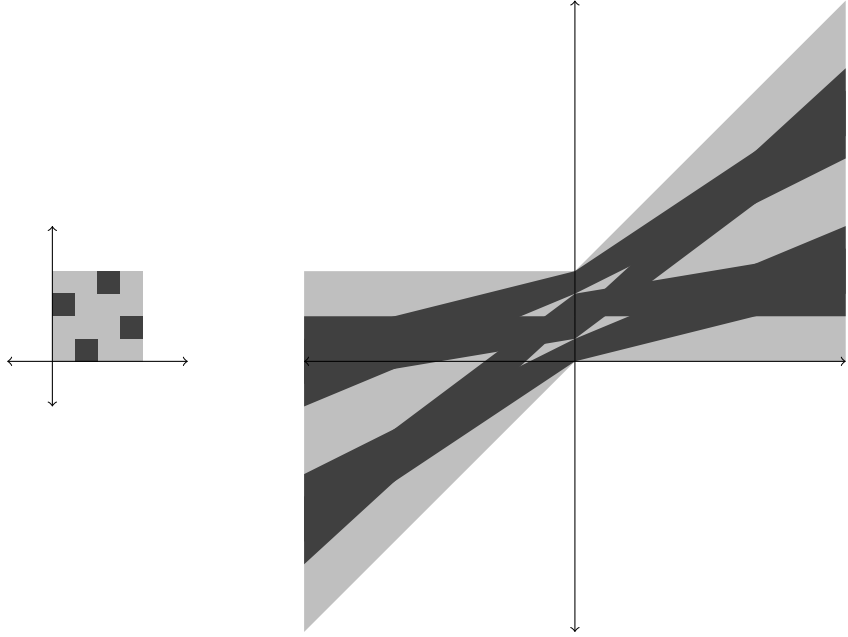

Let $\Ref_Y:\R^2\to\R^2$ and $\Rot_\theta:\R^2\to\R^2$ denote reflection across the $y$-axis and rotation about the origin by the angle $\theta$, respectively. The set
\begin{equation}\label{eq:5}
B=\mathcal{L}\left(F\right)\cup\Rot_{\frac{\pi}{2}}\left(\mathcal{L}\left(F\right)\right)\cup\Ref_Y\bigl(\mathcal{L}\left(F\right)\cup\Rot_{\frac{\pi}{2}}\left(\mathcal{L}\left(F\right)\right)\bigr)
\end{equation}
is the Besicovitch set that we use for our main theorem.
\begin{obser}\label{obs:2}
The set $B$ contains a line in every direction in $\R^2$.
\end{obser}

\begin{proof}
Let $m\in[0,1]$. By \eqref{eq:5} it suffices to show that $\mathcal{L}(F)$ contains a line of slope $m$. But this is clear, since each $F_k$, and hence $F$, contains a point of the form $(m,b)$.
\end{proof}

Using the duality principle and some nontrivial fractal geometry, Falconer also proved the following.
\begin{lemm}\label{lem:3}
\rm(\cite{Falc80,Falc85}) \it The set $B$ has Lebesgue measure $0$.
\end{lemm}

It is not obvious whether or how the proof of Lemma \ref{lem:3} can be effectivized. Nevertheless we prove the following.
\begin{thm}\label{thm:4}\rm(main theorem, in $\R^2$) \it The set $B$ has computable measure $0$. Hence there is, in every direction in $\R^2$, a line that misses every computably random point.
\end{thm}

To prove Theorem \ref{thm:4}, we begin with a property of the line set operator.

\begin{lemm}\label{fourcornerslem}
Let $I$ and $J$ be closed intervals of finite length, so that $R=I\times J$ is a solid rectangle. Let $R^\prime=(I\times J^\circ)\cup(I^\circ\times J)$ be $R$ with its four corners removed. Then
\[\mathcal{L}(R^\prime)\subseteq\mathcal{L}(R)^\circ\cup Y,\]
where $Y=\{(0,y)\;|\;y\in\R\}$ is the $y$-axis of $\R^2$.
\end{lemm}
\begin{proof}
First let $(m,b)\in I\times J^\circ$. Then there exists $\varepsilon>0$ such that
\[\{m\}\times(b-\varepsilon,b+\varepsilon)\subseteq I\times J\;.\]
Then $\mathcal{L}_{m,b^\prime}\subseteq\mathcal{L}(I\times J)=\mathcal{L}(R)$ holds for all $b^\prime\in(b-\varepsilon,b+\varepsilon)$, so $\mathcal{L}_{m,b}\subseteq\mathcal{L}(R)^\circ$. This shows that $\mathcal{L}(I\times J^\circ)\subseteq\mathcal{L}(R)^\circ$.

Now let $(m,b)\in I^\circ\times J$. Then there exists $\varepsilon>0$ such that
\[(m-\varepsilon, m+\varepsilon)\times\{b\}\subseteq I\times J\;.\]
Then $\mathcal{L}_{m^\prime,b}\subseteq\mathcal{L}(I,J)=\mathcal{L}(R)$ holds for all $m^\prime\in(m-\varepsilon,m+\varepsilon)$, so
\[\mathcal{L}_{m,b}\subseteq\mathcal{L}(R)^\circ\cup\{(0,b)\}\subseteq\mathcal{L}(R)^\circ\cup Y\;.\]
This shows that $\mathcal{L}(I^\circ\times J)\subseteq\mathcal{L}(R)\cup Y$.
\end{proof}

Lemma \ref{fourcornerslem} has the following consequence for the stages $F_k$ in the construction of $F$.
\begin{coro}\label{nestingcor}
For every $k\in\N$
\[\mathcal{L}(F_{k+1})\subseteq\mathcal{L}(F_k)^\circ\cup Y.\]
\end{coro}
\begin{proof}
It suffices to note that $F_{k+1}$ does not contain any of the corners of the squares comprising $F_k$.
\end{proof}
\begin{proof}[Proof of Theorem \ref{thm:4}]
Trivial martingale transformations show that the sets of computable measure $0$ in $\R^2$ are closed under $90^\circ$ rotations, reflections about the coordinate axes, and finite unions. Hence by \eqref{eq:5} it suffices to prove that $\mathcal{L}(F)$ has computable measure $0$. We do this by presenting a computable martingale $d$ that succeeds at every point $\mb{x}\in\mathcal{L}(F)^\#$, where
\[\mathcal{L}(F)^\#=[0,1)^2\cap\bigcup_{\mb{t}\in\Z^2}(\mathcal{L}(F)+\mb{t})\]
is defined as in \eqref{eq:2}.

For each $\mb{t}\in\Z^2$ and $k\in\N$ let
\[H_{\mb{t},k}=[0,1)^2\cap\left(\mathcal{L}(F_k)^\circ+\mb{t}\right),\]
noting that $H_{\mb{t},k}$ is an open set in the subspace $[0,1)^2$ of $\R^2$. The sets $F_k$ are so simply defined that the function $h:\Z^2\times\N\to\Q$ defined by
\[h(\mb{t},k)=m(H_{\mb{t},k})\]
is computable. For each $\mb{t}\in\Z^2$ Lemma \ref{lem:3} tells us that
\begin{align*}
0&=m\left([0,1)^2\cap(\mathcal{L}(F)+\mb{t})\right)\\
&=m\left(\bigcap_{k=0}^\infty([0,1)^2\cap(\mathcal{L}(F_k)+\mb{t}))\right)\\
&=\lim_{k\to\infty}m\left([0,1)^2\cap(\mathcal{L}(F_k)+\mb{t})\right)\\
&=\lim_{k\to\infty}h(\mb{t},k).
\end{align*}
Hence the function $k:\Z^2\times\N\to\N$ defined by
\[k(\mb{t},j)=\textrm{the least }k\textrm{ such that }g(\mb{t},k)\leq2^{-j}\]
is also computable.

For each $\mb{t}=(t_1,t_2)\in\Z^2$ and $j\in\N$, define the set $G_{\mb{t},j}$ and the coefficient $c_{\mb{t},j}$ as follows.
\begin{enumerate}[(i)]
\item If $H_{\mb{t},k(\mb{t},|t_1|+|t_2|+j)}\neq\emptyset$, then
\[G_{\mb{t},j}=H_{\mb{t},k(\mb{t},|t_1|+|t_2|+j)}\]
and
\[c_{\mb{t},j}=m(G_{\mb{t},j}).\]
\item Otherwise,
\[G_{\mb{t},j}=[0,1)^2\]
and
\[c_{\mb{t},j}=2^{-(|t_1|+|t_2|+j)}.\]
\end{enumerate}
Define the special-purpose martingale $d_Y$ by
\[d_Y(Q_{r}(\mb{u}))=\left\{\begin{array}{lr}2^r&\textrm{ if }u_1=0\\0&\textrm{ if }u_1>0\end{array}\right.\]
for all $r\in\N$ and $\mb{u}=(u_1,u_2)\in\{0,...,2^r-1\}^2$. Finally, let
\[d=d_Y+\sum_{\mb{t}\in\Z^2}\sum_{j=0}^\infty c_{\mb{t},j}d_{G_{\mb{t},j}}\;,\]
where each $d_{G_{\mb{t},j}}$ is defined from $G_{\mb{t},j}$ as in Section \ref{se:3}. Then
\begin{align*}
d\left([0,1)^2\right)&\leq1+\sum_{\mb{t}\in\Z^2}\sum_{j=0}^\infty2^{-(|t_1|+|t_2|+j)}\\
&=19<\infty,
\end{align*}
so $d$ is a martingale.

To see that $d$ is computable, define
\[\widehat{d}:\N\times J\to\Q\]
(where $J$ is defined as in section 2) by
\[\widehat{d}(s,r,\mb{u})=d_Y(Q_r(\mb{u}))+\sum_{t_1=-p}^p\sum_{t_2=-p}^p\sum_{j=0}^p c_{\mb{t},j}d_{G_{\mb{t},j}}(Q_r(\mb{u})),\]
where $p=s+2r+6$. Then $\widehat{d}$ is computable, and it is clear that $\widehat{d}(s,r,\mb{u})\leq d(Q_r(\mb{u}))$ holds for all $(s,r,\mb{u})\in\N\times J$. We now fix $(s,r,\mb{u})\in\N\times J$, let $p=s+2r+6$, and estimate the difference $d(Q_r(\mb{u}))-\widehat{d}(s,r,\mb{u})$.

For each index set $I\subseteq\Z^2\times\N$ define the sums
\[\sigma(I)=\sum_{(\mb{t},j)\in I}c_{\mb{t},j}d_{G_{\mb{t},j}}(Q_r(\mb{u}))\]
and
\[\tau(I)=\sum_{(\mb{t},j)\in I}2^{-(|t_1|+|t_2|+j)}.\]
By the trivial bound $d_{G_{\mb{t},j}}(Q_r(\mb{u}))\leq4^r$ and the fact that $c_{\mb{t},j}\leq2^{-(|t_1|+|t_2|+j)}$ always holds, we have
\[\sigma(I)\leq4^r\tau(I)\]
for every $I\subseteq\Z^2\times\N$. Now
\[d(Q_r(\mb{u}))=d_Y(Q_r(\mb{u}))+\sigma\left(\Z^2\times\N\right)\;,\]
and
\[\widehat{d}(s,r,\mb{u})=d_Y(Q_r(\mb{u}))+\sigma(I_0),\]
where
\[I_0=\left\{(t_1,t_2,j)\;\big|\;-p\leq t_1\leq p,\;-p\leq t_2\leq p,\;j\leq p\right\}\;,\]
so
\begin{align*}
d(Q_r(\mb{u}))-\widehat{d}(s,r,\mb{u})&=\sigma\left(\left(\Z^2\times\N\right)\smallsetminus I_0\right)\\
&\leq4^r\tau\left(\left(\Z^2\times\N\right)\smallsetminus I_0\right)\;.
\end{align*}
If we let
\[I_a=\left\{(t_1,t_2,j)\;\big|\;|t_a|>p\right\}\]
for $a\in\{1,2\}$ and
\[I^+=\left\{(t_1,t_2,j)\;\big|\;j>p\right\},\]
then
\[\left(\Z^2\times\N\right)\smallsetminus I_0\subseteq I_1\cup I_2\cup I^+,\]
so
\[d(Q_r(\mb{u}))-\widehat{d}(s,r,\mb{u})\leq4^r\left(\tau(I_1)+\tau(I_2)+\tau\left(I^+\right)\right).\]
Now
\begin{align*}
\tau(I_1)&=\tau(I_2)\\
&=2\sum_{t_1=p+1}^\infty2^{-t_1}\sum_{t_2=-\infty}^\infty2^{-|t_2|}\sum_{j=0}^\infty2^{-j}\\
&=12\sum_{t_1=p+1}^\infty2^{-t_1}\\
&=12\cdot2^{-p},
\end{align*}
and
\begin{align*}
\tau(I^+)&=\sum_{t_1=-\infty}^\infty2^{-|t_1|}\sum_{t_2=-\infty}^\infty2^{-|t_2|}\sum_{j=p+1}^\infty2^{-j}\\
&=9\cdot2^{-p},
\end{align*}
so
\begin{align*}
d(Q_r(\mb{u}))-\widehat{d}(s,r,\mb{u})&\leq4^r\cdot33\cdot2^{-p}\\
&=33\cdot2^{-(s+6)}\\
&<2^{-s}\;.
\end{align*}
Hence $\widehat{d}$ testifies that $d$ is computable.

To see that $d$ succeeds at every point in $\mathcal{L}(F)^\#$, let $\mb{x}\in[0,1)^2\cap(\mathcal{L}(F)+\mb{t})$. By Corollary \ref{nestingcor} we have two cases.

Case 1. $\mb{x}\in Y$. Then
\begin{align*}
\limsup_{r\to\infty}d^{(r)}(\mb{x})&\geq\limsup_{r\to\infty}d_Y^{(r)}(\mb{x})\\
&=\limsup_{r\to\infty}2^r\\
&=\infty\;,
\end{align*}
so $d$ succeeds at $\mb{x}$.

Case 2. $\mb{x}\in\mathcal{L}(F_k)$ for every $k\in\N$. Then $\mb{x}\in H_{\mb{t},k}$ for every $k\in\N$, so clause (i) holds in the definitions of $G_{\mb{t},j}$ and $c_{\mb{t},j}$ for every $j\in\N$, with $\mb{x}\in G_{\mb{t},j}$. By Theorem \ref{thm:1}, this implies that
\begin{align*}
\limsup_{r\to\infty}d^{(r)}(\mb{x})&\geq\limsup_{r\to\infty}\sum_{j=0}^\infty c_{\mb{t},j}d_{G_{\mb{t},j}}^{(r)}(\mb{x})\\
&=\limsup_{r\to\infty}\sum_{j=0}^\infty m(G_{\mb{t},j})d_{G_{\mb{t},j}}^{(r)}(\mb{x})\\
&=\infty,
\end{align*}
whence $d$ succeeds at $\mb{x}$.
\end{proof}

In remarks on an early draft of this paper, Turetsky and an anonymous reviewer pointed out an alternative proof of Theorem \ref{thm:4}. The key fact, proved by Wang~\cite{Wang96,DowHir10}, is that every computably random point $\mb{x}$ is \emph{Kurtz random} (also called \emph{weakly random}~\cite{Kurt81}), meaning that $\mb{x}$ is not an element of any computably closed (i.e., $\Pi^0_1$) set of measure 0. Furthermore, the above-mentioned fact that the operator $\mathcal{L}$ maps compact sets to closed sets can be extended to prove that $\mathcal{L}$ maps bounded $\Pi^0_1$ sets to $\Pi^0_1$ sets. Finally, it is routine to verify that the fractal dust $F$ is a bounded $\Pi^0_1$ set. These things and Lemma \ref{lem:3} imply that $\mathcal{L}(F)$ contains no computably random point, whence Theorem \ref{thm:4} holds by Observation \ref{obs:2}. This elegant proof is simpler than our martingale construction, even when Wang's proof is included. However, we believe that the direct martingale construction may help illuminate the path to results on time-bounded randomness, so we retain the martingale proof in this paper.
\end{section}

\begin{section}{Betting in Doubly Exponential Time}\label{se:5}
In light of Theorem \ref{thm:4} it is natural to ask whether there is, in every direction in $\R^2$, a line that misses not only every computably random point, but every feasibly random point. We do not know the answer to this question at the time of this writing, but we prove a weaker result of this type in this section.

As noted in the introduction, Besicovitch constructed a \emph{Kakeya set}, a Lebesgue measure 0 plane set containing a unit line segment in every direction, in 1919. Our objective here is to specify a Kakeya set $K$ and prove that it has ee-measure 0 (a condition defined in Section \ref{se:2}). Our specification and proof take advantage of Besicovitch's original work, together with subsequent refinements by Perron~\cite{Perr28}, Schoenberg~\cite{Scho62}, and Falconer~\cite{Falc85}.

We first describe \emph{Perron trees}, the building blocks of our set $K$. Let
$$\tau=\triangle(\mb{u},\mb{v},\mb{w})$$
be a triangle of height $h$ with its base $\overline{\mb{uv}}$ on the $x$-axis. In this discussion we regard triangles as \emph{including} their interiors. Note that $\tau$ contains a line segment of length $h$ in every direction between the directions of $\overline{\mb{uw}}$ and $\overline{\mb{vw}}$. Given a positive integer $k$, cut $\tau$ into $2^k$ nonoverlapping triangles $\tau_1,...,\tau_{2^k}$ of equal area, as indicated in Figure \ref{fig:2}(a). (Throughout this discussion, sets in $\R^2$ are \emph{nonoverlapping} if their interiors are disjoint.) Besicovitch showed that these $2^k$ smaller triangles can be slid horizontally along the $x$-axis in such a way that their union, due to high overlap, has very small area. Perron simplified Besicovitch's overlap scheme to that depicted in Figure \ref{fig:2}(b). Note that, notwithstanding its small area, the set in Figure \ref{fig:2}(b) still contains a line segment in every direction between the directions of $\overline{\mb{uw}}$ and $\overline{\mb{vw}}$. Schoenberg coined the term \emph{Perron trees} for sets of the type depicted in Figure \ref{fig:2}(b) and gave a simpler, recursive ``sprouting construction'' of the Perron tree $P_k(\tau)$ as a union of $2^{k+1}-1$ \emph{nonoverlapping} triangles as in Figure \ref{fig:2}(c).
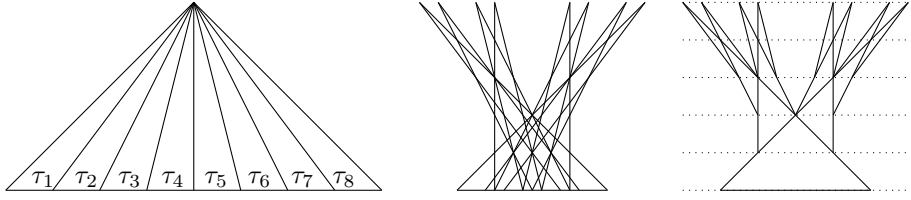
\begin{figure}[h]
	\centering
		\begin{tikzpicture}[scale=0.5]
			\node at (-20,0.3) {$\tau_1$};
			\node at (-18.85,0.3) {$\tau_2$};
			\node at (-17.7,0.3) {$\tau_3$};
			\node at (-16.55,0.3) {$\tau_4$};
			\node at (-15.4,0.3) {$\tau_5$};
			\node at (-14.25,0.3) {$\tau_6$};
			\node at (-13.1,0.3) {$\tau_7$};
			\node at (-12,0.3) {$\tau_8$};
			\draw (-21,0) -- (-16,5);
			\draw (-19.75,0) -- (-16,5);
			\draw (-18.5,0) -- (-16,5);
			\draw (-17.25,0) -- (-16,5);
			\draw (-16,0) -- (-16,5);
			\draw (-14.75,0) -- (-16,5);
			\draw (-13.5,0) -- (-16,5);
			\draw (-12.25,0) -- (-16,5);
			\draw (-11,0) -- (-16,5);
			\draw (-21,0) -- (-11,0);
			
			\draw (-9,0) -- (-4,5);
			\draw (-7.75,0) -- (-4,5);
			
			\draw (-8.25,0) -- (-4.5,5);
			\draw (-7,0) -- (-4.5,5);
			
			\draw (-8,0) -- (-5.5,5);
			\draw (-6.75,0) -- (-5.5,5);
			
			\draw (-7.25,0) -- (-6,5);
			\draw (-6,0) -- (-6,5);
			
			\draw (-8,0) -- (-8,5);
			\draw (-6.75,0) -- (-8,5);
			
			\draw (-7.25,0) -- (-8.5,5);
			\draw (-6,0) -- (-8.5,5);
			
			\draw (-7,0) -- (-9.5,5);
			\draw (-5.75,0) -- (-9.5,5);
			
			\draw (-6.25,0) -- (-10,5);
			\draw (-5,0) -- (-10,5);
			\draw (-9,0) -- (-5,0);

			\draw (-2,0) -- (2,0);
			\draw (-2,0) -- (3,5);
			\draw (2,0) -- (-3,5);
			
			\draw (-1,1) -- (-1,5);
			\draw (1,1) -- (1,5);
			
			\draw (-1,2) -- (-2.5,5);
			\draw (1,2) -- (2.5,5);
			
			\draw (0,2) -- (-1.5,5);
			\draw (0,2) -- (1.5,5);
			
			\draw (-1.5,3) -- (-3,5);
			\draw (1.5,3) -- (3,5);
			
			\draw (-1,3) -- (-2.5,5);
			\draw (1,3) -- (2.5,5);
			
			\draw (-1,3) -- (-1.5,5);
			\draw (1,3) -- (1.5,5);
			
			\draw (-0.5,3) -- (-1,5);
			\draw (0.5,3) -- (1,5);
			
			\draw[dotted] (-3,0) -- (3,0);
			\draw[dotted] (-3,1) -- (3,1);
			\draw[dotted] (-3,2) -- (3,2);
			\draw[dotted] (-3,3) -- (3,3);
			\draw[dotted] (-3,4) -- (3,4);
			\draw[dotted] (-3,5) -- (3,5);
		\end{tikzpicture}
	\caption{(a) a triangle cut into eight pieces $\tau_1,...\tau_8$; (b) a Perron tree constructed by sliding those pieces together; (c) the same Perron tree via Schoenberg's sprouting construction.}
	\label{fig:2}
\end{figure}

We now specify the Perron tree $P_k(\tau)$. Throughout this discussion script capital letters represent collections of nonoverlapping polygons. Let
$$\left(\wh{\mb{u}},\wh{\mb{v}},\wh{\mb{w}}\right)=\left(\frac{2\mb{u}}{h},\frac{2\mb{v}}{h},\frac{2\mb{w}}{h}\right)\;,$$
so that $\bigtriangleup\left(\wh{\mb{u}},\wh{\mb{v}},\wh{\mb{w}}\right)=\frac{2}{h}\tau$ is a triangle of height 2, similar to $\tau$, with its base on the $x$-axis. Let
$$\mathcal{T}_0=\left\{\bigtriangleup\left(\wh{\mb{u}},\wh{\mb{v}},\wh{\mb{w}}\right)\right\}\;,$$
and
$$\mathcal{T}_1=\left\{\bigtriangleup\left(\frac{\wh{\mb{u}}+\wh{\mb{w}}}{2},\wh{\mb{w}},\frac{3\wh{\mb{w}}-\wh{\mb{v}}}{2}\right),\bigtriangleup\left(\frac{\wh{\mb{v}}+\wh{\mb{w}}}{2},\wh{\mb{w}},\frac{3\wh{\mb{w}}-\wh{\mb{u}}}{2}\right)\right\}\;.$$
For $i\geq2$, let
$$\mathcal{T}_i=\bigcup_{t\in\mathcal{T}_{i-1}}\{\bigtriangleup(\mb{m}_t,\mb{c}_t,2\mb{c}_t-\mb{b}_t),\bigtriangleup(\mb{b}_t,\mb{c}_t,2\mb{c}_t-\mb{m}_t)\}\;,$$
where for $t\in\mathcal{T}_{i-1}$, $\mb{a}_t$, $\mb{b}_t$, and $\mb{c}_t$ are the vertices of $t$ with $y$-coordinates $i-1$, $i$, and $i+1$, respectively, and $\mb{m}_t$ is the midpoint of $\overline{\mb{a}_t\mb{c}_t}$. Let
$$\mathcal{T}=\bigcup_{i=0}^k\mathcal{T}_i\;.$$
Then
$$P_k(\tau)=\frac{h}{k+2}\bigcup\mathcal{T}$$
is the $k$-level Perron tree based on $\tau$.

As Schoenberg noted, we can define the same Perron tree as a union of shifted triangles $\tau_i$. Let
\begin{align*}
C&=\left\{x\;\big|\;(x,h)\in P_k(\tau)\right\}\\
&=\left\{x\;\big|\;\mb{c}_t=(x,h)\mbox{ for some }t\in\mathcal{T}_k\right\}\;,
\end{align*}
and index the elements of $C$ as $c_1,...,c_{2^k}$, where $c_i<c_{i+1}$ for $1\leq i<2^k$. Let
$$\mathcal{P}_k(\tau)=\left\{\tau_i+(c_i,0)\;\middle|\;1\leq i<2^k\right\}\;.$$
Then
$$P_k(\tau)=\bigcup\mathcal{P}_k(\tau)\;.$$
\begin{obser}[Falconer~\cite{Falc85}]\label{obs:5.1}
$P_k(\tau)$ is contained in the trapezoid with vertex set
$$\{2\mb{u}-\mb{v},\mb{w}-\mb{v}+\mb{u},\mb{w}+\mb{v}-\mb{u},2\mb{v}-\mb{u}\}\;.$$
\end{obser}
\begin{thm}[Schoenberg~\cite{Scho62}]\label{thm:5.2}
$m(P_k(\tau))=\frac{m(\tau)}{2k+4}$.
\end{thm}

We now construct a sequence $\{S_j\}_{j\in\N}$ of plane sets. Each $S_j$ is the union of a collection $\mathcal{S}_j$ of triangles. Define
$$\mathcal{S}_0=\left\{\bigtriangleup\left((0,0),(1,0),(1/2,1/2)\right)\right\}\;,$$
and for $j\geq1$, cut the base of each triangle in $\mathcal{S}_{j-1}$ into $2^{j+1}$ equal pieces in the manner of Figure \ref{fig:2}(a) to form
$$p_j=2^{j+1}|\mathcal{S}_{j-1}|$$
triangles $\tau_j^{1},...,\tau_j^{p_j}$, and let
$$\mathcal{P}_j^i=\mathcal{P}_{2^j}\left(\tau_j^i\right)\;.$$
Then define
$$\mathcal{S}_j=\bigcup_{i=1}^{p_j}\mathcal{P}_j^i\;,$$
and
$$S_j=\bigcup\mathcal{S}_j\;.$$

For $E\subseteq\R^2$ and $\varepsilon>0$, define
$$G_\varepsilon(E)=\bigcup_{\delta\in(-\varepsilon,\varepsilon)}E^\circ+(\delta,0)\;.$$
We will repeatedly make use of the fact that for any $E_1,E_2\subseteq\R^2$ and $\varepsilon>0$,
$$G_\varepsilon(E_1\cup E_2)=G_\varepsilon(E_1)\cup G_\varepsilon(E_2)\;.$$

For $j\in\N$, let
$$\varepsilon_j=\frac{1}{2^{j+1}|\mathcal{S}_j|}\;,$$
and define
$$G_j=G_{\varepsilon_j}(S_j)\;.$$
\begin{lemm}
For $j\in\N$, $G_{j+1}\subseteq G_j$.
\end{lemm}
\begin{proof}
Since the base of each $\tau_{j+1}^i$ has length $\frac{1}{p_{j+1}}=\varepsilon_j/2$, Observation \ref{obs:5.1} tells us for $i=1,...,p_{j+1}$ that
$$P_{j+1}^i\subseteq\overline{G_{\varepsilon_j/2}\left(\tau_{j+1}^i\right)}\;.$$
Since $\varepsilon_{j+1}<\varepsilon_j/2$, this implies that
$$G_{\varepsilon_{j+1}}\left(P_{j+1}^i\right)\subseteq G_{\varepsilon_j}\left(\tau_{j+1}^i\right)\;.$$
Thus
\begin{align*}
G_{j+1}&=G_{\varepsilon_{j+1}}(S_{j+1})\\
&=G_{\varepsilon_{j+1}}\left(\bigcup_{i=1}^{p_{j+1}}P_{j+1}^i\right)\\
&=\bigcup_{i=1}^{p_{j+1}}G_{\varepsilon_{j+1}}\left(P_{j+1}^i\right)\\
&\subseteq\bigcup_{i=1}^{p_{j+1}}G_{\varepsilon_j}\left(\tau_{j+1}^i\right)\\
&=G_{\varepsilon_j}\left(\bigcup_{i=1}^{p_{j+1}}\tau_{j+1}^i\right)\\
&=G_{\varepsilon_j}(S_j)\\
&=G_j\;.
\end{align*}
\end{proof}

Let
$$F=\bigcap_{j\in\N}\overline{G_j}$$
and
$$K_0=\bigcup_{c\in\N}cF\;.$$
Then let
$$K=K_0\cup\Rot_{\pi/4}(K_0)\;.$$
\begin{prop}
$K$ contains arbitrarily long line segments in every direction in $\R^2$.
\end{prop}
\begin{proof}
It suffices to show that $F$ contains a closed segment of length $\frac13$ in every direction $\theta\in\left[\frac{\pi}{4},\frac{3\pi}{4}\right]$. Fix $\theta\in\nobreak\left[\frac{\pi}4,\frac{3\pi}4\right]$. For each $j\in\N$, fix a closed segment $L_j\subseteq\overline{G_j}$ of length $\frac13$ in direction $\theta$. By compactness there is an infinite set $I\subseteq\N$ such that the sequence $(L_i\;|\;i\in I)$ converges (in Hausdorff distance) to a segment $L$. It is clear that $L$ is a segment of length $\frac13$ in direction $\theta$. Using compactness again, we have $L\subseteq\overline{G_j}$ for each $j$, whence $L\subseteq F$.
\end{proof}
\begin{thm}\label{thm:ee}
The set $K$ has {\normalfont ee}-measure 0. Hence there exist, in every direction in $\R^2$, arbitrarily long line segments that miss every {\normalfont ee}-random point.
\end{thm}
\begin{proof}
We first show by induction that $|\mathcal{S}_j|\leq 2^{2^{j+2}}$. This holds for $j=0$. Fix $j\geq 1$ and suppose the claim holds for $j-1$. Then we have
\begin{align*}
|\mathcal{S}_j|&=\sum_{i=1}^{p_j}\left|\mathcal{P}_j^i\right|\\
&=p_j\cdot2^{2^j}\\
&=2^{j+1}|\mathcal{S}_{j-1}|\cdot 2^{2^j}\\
&\leq2^{j+1+2^{j+1}+2^j}\\
&\leq2^{2^{j+2}}\;,
\end{align*}
so the claim holds for every $j\in\N$.

Now consider $G_j^i=G_{\varepsilon_j}\left(P_j^i\right)$. By Theorem \ref{thm:5.2},
\begin{align*}
m\left(P_j^i\right)&=\frac{1}{(4p_j)\left(2\cdot2^j+4\right)}\\
&<\frac{1}{2^{j+3}p_j}\;,
\end{align*}
and
\begin{align*}
m\left(G_j^i\smallsetminus P_j^i\right)&\leq\sum_{t\in\mathcal{P}_j^i}m\left(G_{\varepsilon_j}(t)\smallsetminus t\right)\\
&\leq\left|\mathcal{P}_j^i\right|\varepsilon_j\\
&=\frac{|\mathcal{S}_j|\varepsilon_j}{p_j}\\
&=\frac{1}{2^{j+1}p_j}\;,
\end{align*}
so
$$m\left(G_j^i\right)<\frac{1}{2^j p_j}\;.$$

Define the martingale $d:\mathcal{Q}\to[0,\infty)$ by
$$d(Q)=\sum_{j=0}^\infty2^{-j}\sum_{i=1}^{p_j}\frac{d_{G_j^i}(Q)}{p_j}\;,$$
where $d_{G_j^i}$ is the open set martingale, as defined in Section \ref{se:3}, for $G_j^i$. Then by Theorem \ref{thm:1}, for every $\mb{x}\in G_j^i$,
\begin{align*}
\lim_{r\to\infty}d_{G_j^i}^{(r)}(\mb{x})&=\frac{1}{m\left(G_j^i\right)}\\
&> 2^j p_j\;.
\end{align*}
Thus since
$$G_j=\bigcup_{i=1}^{p_j}G_j^i\;,$$
we have
$$\limsup_{r\to\infty}d^{(r)}(\mb{x})\geq\left|\left\{j \;|\;\mb{x}\in G_j\right\}\right|\;,$$
which is infinite for $\mb{x}\in F$, i.e., $d$ succeeds on every $\mb{x}\in F$.

We now turn to proving that $d$ is ee-computable. Let $J$ be as in Section \ref{se:2}, and define the function $$\wh{d}:\N\times J\to\Q\cap[0,\infty)$$ by
$$\wh{d}(s,r,\mb{u})=\sum_{j=0}^{2r+s}2^{-j}\sum_{i=1}^{p_j}\frac{d_{G_j^i}(Q_r(\mb{u}))}{p_j}\;.$$
Then
\begin{align*}
\left|d(Q_r(\mb{u}))-\wh{d}(s,r,\mb{u})\right|&=\sum_{j=2r+s+1}^{\infty}2^{-j}\sum_{i=1}^{p_j}\frac{d_{G_j^i}(Q_r(\mb{u}))}{p_j}\\
&\leq4^r\sum_{j=2r+s+1}^{\infty}2^{-j}\\
&=2^{-s}\;.
\end{align*}

It remains to be shown that $\wh{d}(s,r,\mb{u})$ is computable in time $2^{2^{O(r+s)}}$. For this it is to show that $d_{G_j^i}(Q_r(\mb{u}))$ is computable in time $2^{2^{O(r+s)}}$ for each $1\leq i\leq p_j$ and $0\leq j\leq 2r+s$. By (\ref{eq:closeddg}),
$$d_{G_j^i}(Q_r(\mb{u}))=4^r\frac{m\left(G_j^i\cap Q_r(\mb{u})\right)}{m\left(G_j^i\right)}\;.$$
Now
\begin{align*}
G_j^i&=G_{\varepsilon_j}\left(P_j^i\right)\\
&=G_{\varepsilon_j}\left(\bigcup\mathcal{P}_j^i\right)\\
&=\bigcup_{t\in\mathcal{P}_j^i}G_{\varepsilon_j}(t)
\end{align*}
is the union of $\nu=2^{2^j}$ trapezoids. The vertices of the triangles $t$ come from the definition of $\mathcal{P}_j(t)$, and the vertices of the trapezoid follow immediately. This gives $4\nu$ segments whose intersections can be found efficiently. In total, the figure $G_j^i$ has at most $4\nu+\binom{4\nu}{2}=poly(\nu)$ vertices, each of which can be found in $poly(\nu)$ time. So it can be triangulated in $poly(\nu)$ time into a set $\mathcal{N}$ of nonoverlapping triangles with $|\mathcal{N}|=poly(\nu)$, where the vertices of every triangle in $\mathcal{N}$ an be found in $poly(\nu)$ time. Then
$$m\left(G_j^i\right)=\sum_{t\in\mathcal{N}}m(t)\;,$$
and for any $Q\in\mathcal{Q}$,
$$m\left(G_j^i\cap Q\right)=\sum_{t\in\mathcal{N}}m(t\cap Q)\;.$$
Hence we can compute $d_{G_j^i}$ in time $poly(\nu)=2^{2^{O(r+s)}}$.
\end{proof}
\end{section}

\begin{section}{Higher Dimensions}\label{se:6}
For every $n\in\N$, the set $B\times\R^n$ contains a line in every direction in $\R^{n+2}$, and Fubini's theorem implies that this set has Lebesgue measure $0$~\cite{Falc03}. In this section we show that $B\times\R^n$ also has computable measure $0$.

For any set $E\subseteq\R^n$ and $\mb{y}\in\R^m$, for $1\leq m<n$, define
\[E_\mb{y}=\left\{\left(x_1,...,x_{n-m}\right)\in\R^{n-m}\;|\;\left(x_1,...,x_{n-m},y_1,...,y_m\right)\in E\right\}\;.\]
The following computable Fubini theorem may be known, but we do not know a reference at the time of this writing.
\begin{thm}\label{thm:7}
Let $E\in\R^n$. If there is a computable martingale $d$ on $\left[0,1\right)^{n-m}$ such that the set
\[N_E\left(d\right)=\left\{\mb{y}\in\left[0,1\right)^m\;|\;\exists\:\mb{x}\in E^\#_\mb{y}\mbox{ such that }d\mbox{ does not succeed at }\mb{x}\right\}\]
has computable measure $0$, then $E$ has computable measure $0$.
\end{thm}
\begin{proof}
Let $d_1$ be such a martingale for $E$ and let $d_2$ be a computable martingale on $[0,1)^m$ that succeeds at every $\mb{y}\in N_E(d_1)$. Define two martingales on $[0,1)^n$, $d_1^\prime$ and $d_2^\prime$, by
\[d_1^\prime(Q_r(u_1,...,u_n))=d_1(Q_r(u_1,...,u_{n-m}));\]
\[d_2^\prime(Q_r(u_1,...,u_n))=d_2(Q_r(u_{n-m+1},...,u_n)).\]
Note that both are computable.

Now let $\mb{x}=(x_1,...,x_n)\in E^\#$. If $(x_{n-m+1},...,x_n)\in N_E(d_1)$, then $d_2^\prime$ succeeds at $\mb{x}$; otherwise, $d_1^\prime$ succeeds at $\mb{x}$. We conclude that the computable martingale $d=d_1^\prime+d_2^\prime$ succeeds at every $\mb{x}\in E^\#$, whence $E$ has computable measure $0$.
\end{proof}
\begin{coro}\label{cor:8}
For every computable measure $0$ set $E$ and $n\in\N$, the set $E\times\R^n$ has computable measure $0$.
\end{coro}
\begin{thm}\label{thm:9}
\rm(main theorem, in $\R^n$) \it For every $n\geq2$ there is, in every direction in $\R^n$, a line that misses every computably random point.
\end{thm}
\begin{proof}
By Theorem \ref{thm:4}, $B$ has computable measure $0$. Thus by Corollary \ref{cor:8}, $B\times\R^{n-2}$ has computable measure $0$ for every $n\geq 3$.
\end{proof}
It is routine to prove double exponential time versions of Theorem \ref{thm:7} and Corollary \ref{cor:8}, and hence to extend Theorem \ref{thm:ee} to $\R^n$ in a similar fashion.
\end{section} 

\begin{section}{Open Problems}\label{se:7}
As noted in the introduction, we conjecture that there is a line in every direction missing every feasibly random point in Euclidean space. Proving or disproving this conjecture may require a significant advance beyond current understanding of the algorithmic geometric measure theory of Besicovitch and Kakeya sets. In the meantime, more modest goals may be achieved. Can Theorem \ref{thm:ee} be improved to singly exponential time, or to lines instead of segments?

Besicovitch's duality idea for constructing the set $B$ came soon after, and was perhaps prompted by, the Mathematical Association of America's production of a film in which he explained his 1919 solution of the Kakeya needle problem. (The article~\cite{Besi63} is based on this film.) Does a copy of this film still exist?
\end{section}
\subsection*{Acknowledgment} We thank Dan Turetsky and an anonymous reviewer for pointing out the alternate proof of Theorem \ref{thm:4} and for a useful correction.

\bibliography{lmerp}
\bibliographystyle{abbrv}

\end{document}